\documentclass[11pt]{article}

\usepackage[utf8]{inputenc}
\usepackage[T1]{fontenc}
\usepackage[english]{babel}
\usepackage{fullpage}

\usepackage{float}
\usepackage{amsthm}
\usepackage{amsmath}
\usepackage{amssymb}
\usepackage{mathrsfs}
\usepackage{tikz}

\theoremstyle{plain}
\newtheorem{theorem}{Theorem}
\newtheorem{proposition}[theorem]{Proposition}
\newtheorem{corollary}[theorem]{Corollary}
\newtheorem{lemma}[theorem]{Lemma}

\newtheorem{definition}{Definition}
\newtheorem{example}{Example}

\newcommand{\R}{\mathbb{R}}
\newcommand{\Z}{\mathbb{Z}}
\newcommand{\F}{\mathbb{F}}

\def\poly{\operatorname{poly}}
\def\polylog{\operatorname{polylog}}
\def\polyloglog{\operatorname{polyloglog}}

\newcommand{\andd}{\textsc{And}}
\newcommand{\orr}{\textsc{OR}}
\newcommand{\nott}{\textsc{Not}}
\newcommand{\nor}{\textsc{Nor}}
\newcommand{\modu}{\textsc{Mod}}
\newcommand{\gip}{\textsc{Gip}}
\newcommand{\maj}{\textsc{Maj}}
\newcommand{\disj}{\textsc{Disj}}
\newcommand{\thr}{\textsc{Thr}}
\newcommand{\equ}{\textsc{Equality}}

\newcommand{\any}{\mathsf{ANY}}
\newcommand{\sym}{\mathsf{SYM}}
\newcommand{\comp}{\mathsf{COMP}}
\newcommand{\dcc}[1]{\mathsf{D}_{#1}}

\newcommand{\acc}{\mathsf{ACC}}

\newcommand{\nexp}{\mathsf{NEXP}}

\newcommand{\bo}[1]{\mathcal{O}\left(#1\right)}
\newcommand{\bomega}[1]{\Omega\left(#1\right)}
\newcommand{\rn}{\{0,1\}}
\newcommand{\ra}{\rightarrow}

\newcommand{\abs}[1]{\left|{#1}\right|}
\newcommand{\vect}[1]{\overrightarrow{#1}}
\newcommand{\sor}[1]{\mathscr{S}{#1}}
\newcommand{\kp}{k^+}

\title{\bf Simultaneous Multiparty Communication Complexity of Composed Functions}

\author{Yassine Hamoudi \\ IRIF, Université Paris Diderot, France. \\ \texttt{hamoudi@irif.fr}}
\date{}

\begin{document}

\maketitle

\begin{abstract}

  The Number On the Forehead (NOF) model is a multiparty communication game between $k$ players that collaboratively want to evaluate a given function $F : \mathcal{X}_1 \times \cdots \times \mathcal{X}_k \ra \mathcal{Y}$ on some input $(x_1,\dots,x_k)$ by broadcasting bits according to a predetermined protocol. The input is distributed in such a way that each player $i$ sees all of it except $x_i$ (as if $x_i$ is written on the forehead of player $i$). In the Simultaneous Message Passing (SMP) model, the players have the extra condition that they cannot speak to each other, but instead send information to a referee. The referee does not know the players' inputs, and cannot give any information back. At the end, the referee must be able to recover $F(x_1,\dots,x_k)$ from what she obtained from the players.

  A central open question in the simultaneous NOF model, called the \emph{$\log n$ barrier}, is to find a function which is hard to compute when the number of players is $\polylog(n)$ or more (where the $x_i$'s have size $\poly(n)$). This has an important application in circuit complexity, as it could help to separate $\acc^0$ from other complexity classes \cite{HG91,BGKL04}. One of the candidates for breaking the $\log n$ barrier belongs to the family of \emph{composed functions}. The input to these functions in the $k$-party NOF model is represented by a $k \times (t \cdot n)$ boolean matrix $M$, whose row $i$ is the number $x_i$ on the forehead of player $i$ and $t$ is a block-width parameter. A symmetric composed function acting on $M$ is specified by two symmetric $n$- and $kt$-variate functions $f$ and $g$ (respectively), that output $f \circ g(M) = f(g(B_1),\dots,g(B_n))$ where $B_j$ is the $j$-th block of width $t$ of $M$. As the majority function $\maj$ is conjectured to be outside of $\acc^0$, Babai \emph{et. al.} \cite{BKL95,BGKL04} suggested to study the composed function $\maj \circ \maj_t$, with $t$ large enough, for breaking the $\log n$ barrier (where $\maj_t$ outputs 1 if at least $kt/2$ bits of the input block are set to 1).

  So far, it was only known that block-width $t = 1$ is not enough for $\maj \circ \maj_t$ to break the $\log n$ barrier in the simultaneous NOF model \cite{BGKL04} (Chattopadhyay and Saks \cite{CS14} found an efficient protocol for $t \leq \polyloglog(n)$, but it requires randomness to be simultaneous). In this paper, we extend this result to any constant block-width $t > 1$ by giving a \emph{deterministic simultaneous} protocol of cost $2^{\bo{2^t}} \log^{2^{t+1}}(n)$ for \textit{any} symmetric composed function $f \circ g$ (which includes $\maj \circ \maj_t$) when there are more than $2^{\Omega(2^t)} \log n$ players.

  \bigbreak\noindent{\bf Keywords:} Communication complexity, Number On the Forehead model, Simultaneous Message Passing, Log n barrier, Symmetric Composed functions.
\end{abstract}


\section{Introduction}
\label{Sec:intro}


  \subsection{Number On the Forehead and Simultaneous models}
  \label{Sec:nof}

The \textit{Number On the Forehead} (NOF) model is a multiparty communication model introduced by Chandra, Furst and Lipton \cite{CFL83} that generalizes the two player communication game of Yao \cite{Yao79}. In this model, $k$ players are given $k$ inputs $x_1 \in \mathcal{X}_1, \dots, x_k \in \mathcal{X}_k$ on which they want to compute some function $F : \mathcal{X}_1 \times \cdots \times \mathcal{X}_k \ra \mathcal{Y}$. Each player $i$ sees all of the input $(x_1,\dots,x_k)$, except $x_i$. The situation is as if input $x_i$ is written on the forehead of player $i$.

In order to collaboratively evaluate $F(x_1,\dots,x_k)$, the players communicate by \textit{broadcasting} bits according to a predetermined \textit{protocol}. This protocol specifies whose turn it is to speak, and which bit is to be sent given the information exchanged so far and the input seen by the speaking player. It also determines when communication stops. At the end, all the players must be able to recover $F(x_1,\dots,x_k)$ from the input they see and the transcript of the exchange. The cost of the protocol on input $(x_1,\dots,x_k)$ is the number of exchanged bits, and the total cost is the worst case cost on all inputs. The $k$-party deterministic communication complexity of $F$, denoted $\dcc{k}(F)$, is the cost of the most efficient protocol computing $F$.

In most of the settings, the $x_i$'s are $\poly n$-bits long (for some parameter $n$) and $\mathcal{Y} = \rn$. In this case, the naive protocol is to broadcast first the entire input $x_1$ (this can be done by player 2), and then player 1 computes $F(x_1,\dots,x_k)$ and sends the result to the other players. This protocol has cost $m+1$ (where $m = \poly(n)$ is the number of bits required for sending $x_1$), which proves $\dcc{k}(F) = \bo{\poly n}$. Consequently, a protocol will be said to be \textit{efficient} if it has cost $\bo{\polylog n}$ (i.e. we seek for exponential speed-up over the naive protocol).

Among the many variants of the previous framework (randomized, quantum, etc.), we will be interested in the \emph{simultaneous} (or \emph{Simultaneous Message Passing - SMP}) model \cite{Yao79,Nis93,BKL95,PRS97} in which the players cannot speak to each other but instead send information to a referee. The referee does not know the players' inputs, and cannot give any information back. At the end, the referee must be able to recover $F(x_1,\dots,x_k)$ from what she obtained from the players. The simultaneous deterministic communication complexity is denoted $\dcc{k}^{||}(F)$, and it always satisfies $\dcc{k}(F) \leq \dcc{k}^{||}(F)$. It has often been easier to reason first in this weaker model for proving lower bounds \cite{BGKL04,PRS97,BPSW05,BJKS02}. It is also more suitable and fruitful for studying certain functions, such as $\equ$ in the two party setting \cite{Yao79,Amb96,NS96,BK97,BCWW01,GRW08,BGK15}. We will show in the next section that the simultaneous deterministic communication model is also closely connected to lower bound results in the complexity class $\acc^0$.



  \subsection[The log n barrier problem and ACC0 lower bounds]{The log n barrier problem and $\acc^0$ lower bounds}
  \label{Sec:lognBarrier}

The NOF model has proved to be of value in the study of many areas of computer science, such as branching programs \cite{CFL83}, Ramsey theory \cite{CFL83}, circuit complexity \cite{HG91, BT94}, quasirandom graphs \cite{CT93}, proof complexity \cite{BPS07}, etc. One of the most interesting connections, pointed out by H{\aa}stad and Goldmann \cite{HG91} and refined in \cite{BGKL04}, is a way to derive lower bounds for the complexity class\footnote{$\acc_0$ refers to the functions computable by constant-depth poly-size circuits with unbounded fan-in $\andd$, $\orr$, $\nott$ and $\modu_m$ gates (where $\modu_m$ outputs $0$ iff the sum of its inputs is divisible by $m$).} $\acc^0$ from lower bounds in the \emph{simultaneous} NOF model. More precisely, according to a result from Yao, Beigel and Tarui \cite{Yao90,BT94}, any function $f \in \acc^0$ can be expressed as a depth-2 circuit whose top gate is a symmetric gate of fan-in $2^{\log^c n}$, and each bottom gate is an $\andd$ gate of fan-in $\log^d n$ (for some constants $c,d$). Consequently, for any partition of the input of $f$ between $k = \log^d n - 1$ players in the \emph{simultaneous} NOF model, there exists a partition of the $\andd$ gates between the players such that each of them sees all the input bits she needs to evaluate the gates she received. The players can then send to the referee the number of gates that evaluate to 1, which enables the referee to compute $f$. The total cost of this protocol is $\bo{k \log\left(2^{\log^c n}\right)} = \bo{\log^{c+d} n}$. Conversely, any super-polylogarithmic lower bound in the \emph{simultaneous} NOF model for a function $f$ and a partition of its input between $\polylog(n)$ players would imply $f \notin \acc^0$.

Separating $\acc^0$ from other complexity classes is a central question in complexity theory. It is conjectured that $\acc^0$ does not contain the majority function $\maj$, but the only result known so far is $\nexp \not\subset \acc^0$ \cite{Wil14}. The aforementioned connection with communication complexity has motivated the search for a function which is hard to compute for $k \geq \log n$ players in the simultaneous NOF model. This problem is called the \textit{$\log n$ barrier}.

Obtaining lower bounds in the NOF model is a challenging task, as the current methods become very weak when $k \geq \log n$. The only general lower bound technique known so far is the discrepancy method and its variants \cite{BNS92,CT93,Raz00,She11}. One of the early application of it was an $\Omega(n/4^k)$ lower bound on the randomized complexity of the \textit{Generalized Inner Product} ($\gip$) function \cite{BNS92}. A long series of generalizations and improvements of the discrepancy method subsequently led to an $\bomega{\frac{\sqrt{n}}{k2^k}}$ (resp. $\Omega(n/4^k)$) lower bound on the randomized (resp. deterministic) complexity of the \textit{Disjointness} ($\disj$) function \cite{Tes03,BPS06,CA08,LS09,BH12,She16,She14,RY15}. It might seem like other lower bound arguments could prove that $\gip$ and $\disj$ remain hard for $k \geq \log n$ players. However, surprising non-simultaneous \cite{Gro94,ACFN15} and simultaneous \cite{BGKL04,ACFN15} protocols proved that the aforementioned lower bounds are nearly optimal, and that these two functions cannot break the $\log n$ barrier. Very recently, Podolskii and Sherstov \cite{PS17} showed that the randomized complexity of $\gip$ and $\disj$ is exactly $\Theta\left(\frac{\log n}{\left\lceil 1 + k/\log n \right\rceil}+1\right)$ when $k \geq \log n$, and built a function having complexity $\Omega(\log n)$ independently of $k$. Although these last results do not break the $\log n$ barrier, they are the first superconstant lower bounds proved for explicit functions when $k \geq \log n$.


  \subsection{Composed Functions}
  \label{Sec:composedFun}

An input $x_1, \dots, x_k \in \rn^n$ to $k$ players in the NOF model can be visualized as a $k \times n$ boolean matrix $M$ where row $i$ is the number $x_i$ on the forehead of player $i$. The protocols known so far for $\gip$ and $\disj$ strongly rely on the particular way these functions act on matrix $M$. They both consist in applying the $g = \andd$ function on each of the $n$ columns of $M$, followed by the $f = \modu_2$ (for $\gip$) or $f = \nor$ (for $\disj$) function on the $n$ resulting bits. Since $\gip$ and $\disj$ do not break the $\log n$ barrier, a natural move has been to try other $f$ and $g$ functions, and to increase the number $t$ of columns on which each $g$ function applies. These are called the \textit{composed functions}, formally defined below and depicted in Figure \ref{Fig:compFun}.

\begin{definition}[Boolean input version]
  \label{Def:compFunBool}
  Fix a block-width parameter $t \geq 1$, and consider functions $f : \rn^n \ra \rn$ and $\vec{g}=(g_1,\dots,g_n)$ where $g_j : (\rn^t)^k \ra \rn$. Given $x_1, \dots, x_k \in \rn^{t \cdot n}$, the \textit{composed function} $f \circ \vec{g}$ for $k$ players outputs $f \circ \vec{g} (x_1, \dots, x_k) = f(g_1(B_1), \dots, g_n(B_n))$ where $B_j \in (\rn^t)^k$ is the $j^{th}$ block of width $t$ in the matrix representation $M$ of the input. When $g = g_1 = \dots = g_n$, we denote $f \circ \vec{g}$ by $f \circ g$.
\end{definition}

\begin{figure}
  \centering
    \begin{tikzpicture}
  \draw (10,10) -- (17,10) ;
  \draw (10,7) -- (17,7) ;

  \draw (11.6,9.5) -- (12.3,9.5) ;
  \draw (10,9.5) -- (10.7,9.5) ;
  \draw (10,9) -- (10.7,9) ;
  \draw (11.6,9) -- (12.3,9) ;
  \draw (10,7.5) -- (10.7,7.5) ;
  \draw (11.6,7.5) -- (12.3,7.5) ;

  \draw (16.2,9.5) -- (17,9.5) ;
  \draw (16.2,9) -- (17,9) ;
  \draw (16.2,7.5) -- (17,7.5) ;
  \draw (16.2,9.5) -- (17,9.5) ; 

  \draw (10,10) -- (10,7) ;
  \draw (17,10) -- (17,7) ;

  \draw (10.7,10) -- (10.7,9) ;
  \draw (10.7,7.5) -- (10.7,7) ;
  \draw (11.6,10) -- (11.6,9) ;
  \draw (11.6,7.5) -- (11.6,7) ;
  \draw (12.3,10) -- (12.3,9) ;
  \draw (12.3,7.5) -- (12.3,7) ;
  \draw (16.2,10) -- (16.2,9) ;
  \draw (16.2,7.5) -- (16.2,7) ;        

  
  \draw[dashed] (10.7,9) -- (10.7,7.5) ;
  \draw[dashed] (11.6,9) -- (11.6,7.5) ;
  \draw[dashed] (12.3,9) -- (12.3,7.5) ;
  \draw[dashed] (16.2,9) -- (16.2,7.5) ;
                
  \draw (11.2,8.3) node{$\cdots$} ;                
  \draw (14.3,8.3) node{$\cdots$} ;                
  \draw (8.4,8.4) node{$\vdots$} ;                

  \draw (10.35,9.7) node{\small $x_{1,1}$} ;                
  \draw (11.95,9.7) node{\small $x_{1,t}$} ;                
  \draw (10.35,9.2) node{\small $x_{2,1}$} ;                
  \draw (11.95,9.2) node{\small $x_{2,t}$} ;  
  \draw (10.35,7.2) node{\small $x_{k,1}$} ;                
  \draw (11.95,7.2) node{\small $x_{k,t}$} ; 
  \draw (16.6,9.7) node{\small $x_{1,tn}$} ;                
  \draw (16.6,9.2) node{\small $x_{2,tn}$} ;                
  \draw (16.6,7.2) node{\small $x_{k,tn}$} ;             

  \draw (8.7,9.75) node{\small Player 1 ($x_1$)} ;  
  \draw (8.7,9.25) node{\small Player 2 ($x_2$)} ;  
  \draw (8.7,7.25) node{\small Player $k$ ($x_k$)} ;  

  \draw[<->] (10,10.2) -- (17,10.2) ;
  \draw (13.5,10.4) node{\small $t \cdot n$} ;  
  \draw[<->] (17.2,10) -- (17.2,7) ;
  \draw (17.4,8.5) node{\small $k$} ; 
    
  \draw (10,6.8) -- (12.3,6.8) ;
  \draw (10,6.8) -- (10,6.9) ;
  \draw (12.3,6.8) -- (12.3,6.9) ;
  \draw[->] (11.05,6.8) -- (11.05,6.5) ;
  \draw (11.05,6.3) node{\small $g_1$} ;

  \draw (14.7,6.8) -- (17,6.8) ;
  \draw (14.7,6.8) -- (14.7,6.9) ;
  \draw (17,6.8) -- (17,6.9) ;
  \draw[->] (15.85,6.8) -- (15.85,6.5) ;
  \draw (15.85,6.3) node{\small $g_n$} ;

  \draw[->] (17.1,6.4) -- (17.6,6.4) ;
  \draw (17.8,6.4) node{\small $f$} ;
\end{tikzpicture}
  \caption{Matrix structure of a composed function $f \circ \protect\vec{g}$ of block-width $t$.}
  \label{Fig:compFun}
\end{figure}
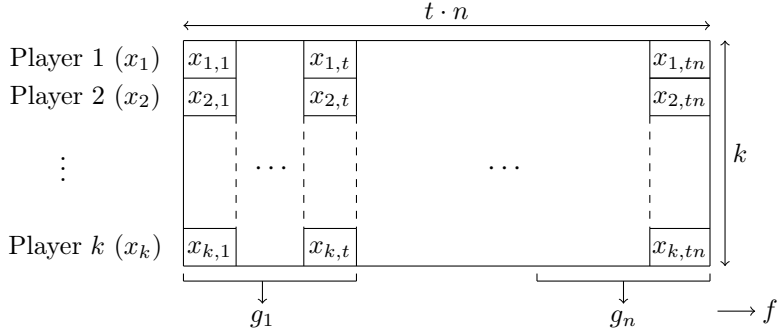

Both $\gip = \modu_2 \circ \andd$ and $\disj = \nor \circ \andd$ are composed functions for $t=1$, with the additional property that $\modu_2$, $\nor$ and $\andd$ are \emph{symmetric} functions (i.e. invariant under any permutation of their input). Since the majority function $\maj$ is conjectured to be outside of $\acc^0$, Babai \emph{et. al.} \cite{BKL95,BGKL04} suggested to look at $\maj \circ \maj_t$ and $\maj \circ \thr^s_t$ for breaking the $\log n$ barrier (where $\maj_t$ outputs 1 if at least $kt/2$ bits of the input block are set to 1, and $\thr^s_t(r_1,\dots,r_k) = 1$ if $r_1 + \dots + r_k \geq s$ for $r_1, \dots, r_k$ seen as $t$-bits numbers).


Another way to look at composed functions of block-width $t$ is to interpret each sub-row $r \in \rn^t$ of each block as a number in $\Z_d$, where $d=2^t$. This representation of the input as a $k \times n$ matrix $M$ over some set $\Z_d$ is sometimes more convenient to use. Below, we reformulate Definition \ref{Def:compFunBool} using this point of view.

\begin{definition}[Integer input version]
  \label{Def:compFun}
  Fix an integer $d \geq 2$ and consider functions $f : \rn^n \ra \rn$ and $\vec{g}=(g_1,\dots,g_n)$ where $g_j : \Z_d^k \ra \rn$. Given $x_1, \dots, x_k \in \Z_d^n$, the \textit{composed function} $f \circ \vec{g}$ for $k$ players outputs $f \circ \vec{g} (x_1, \dots, x_k) = f(g_1(C_1), \dots, g_n(C_n))$ where $C_j \in \Z_d^k$ is the $j^{th}$ column in the matrix representation $M$ of the input. When $g = g_1 = \dots = g_n$, we denote $f \circ \vec{g}$ by $f \circ g$.
\end{definition}

The set of all composed functions $f \circ \vec{g}$ (resp. $f \circ g$) over $\Z_d$ is denoted $\any \circ \vect{\any}_{\Z_d}$ (resp. $\any \circ \any_{\Z_d}$). Similarly, $\sym \circ \sym_{\Z_d}$ is the set of $f \circ g$ for symmetric $f$ and symmetric $g$ functions, $\sym \circ \vect{\any}_{\Z_d}$ is the set of $f \circ \vec{g}$ for symmetric $f$ and any $\vec{g}$, etc. If $d = 2$ (which corresponds to block-width $t = 1$), we will drop the subscript and write $\any \circ \vect{\any}$, $\sym \circ \sym$, etc. We have for instance $\gip, \disj \in \sym \circ \sym$ and $\maj \circ \maj_t, \maj \circ \thr^s_t \in \sym \circ \sym_{\Z_{2^t}}$.

The first efficient protocol for composed functions with $\polylog(n)$ or more players was given by Grolmusz \cite{Gro94}. It is a \emph{non-simultaneous} protocol of cost $\bo{\log^2 n}$ for any composed function in $\sym \circ \andd$ (the inner function is fixed to be $\andd$) when $k \geq \log n$. The study of composed functions with symmetric outer function $f$ was subsequently continued, as it captures many other interesting cases in communication complexity. Babai \emph{et. al.} \cite{BKL95} proposed first $\maj \circ \maj_1$ as a candidate to break the $\log n$ barrier. However, in a subsequent work \cite{BGKL04}, they found a \emph{simultaneous} protocol that applies to $\sym \circ \comp^c$ (where $\comp^c$ holds for \textit{$c$-compressible symmetric} functions\footnote{A class $\mathscr{G}$ (parameterized by $k$) of symmetric functions $g : \rn^k \ra \rn$ is $c$-compressible if for any function $g \in \mathscr{G}$, set $S \subsetneq \{1,\dots,k\}$ and input $(x_i)_{i \in S} \in \rn^{|S|}$ there is a message $m_S$ of size $\bo{1} + c \log(k-|S|)$ such that $g(x_1,\dots,x_k)$ can be computed for any $(x_i)_{i \in \{1,\dots,k\} \backslash S} \in \rn^{k-|S|}$ from knowledge of $m_S$ and $(x_i)_{i \in \{1,\dots,k\} \backslash S}$. The $\maj_1$ and $\thr^s_1$ functions are $1$-compressible \cite{BGKL04}.}, a subset of $\sym$ that contains $\maj$ and $\andd$). It has cost $\bo{\log^{2+c} n}$ when $k > 1 + \log n$. Later, Ada \emph{et. al.} \cite{ACFN15} generalized this result to $\sym \circ \vect{\any}$, with a \emph{simultaneous} protocol of cost $\bo{\log^3 n}$ for $k > 1 + 2 \log n$ players. The only protocol known so far for block-width $t > 1$ has been discovered by Chattopadhyay and Saks \cite{CS14}. It has cost $\bo{d \log n \log(dn)}$ for $\sym \circ \vect{\any}_{\Z_d}$ when $k > 1 + d \log(3n)$ (which is efficient for $d \leq \polylog n$). However, it is \emph{not simultaneous} in the deterministic setting (the authors showed how to make it simultaneous using shared randomness between the players). Thus, none of these previous results prevents from breaking the $\log n$ barrier in the SMP model with composed functions of block-width as small as $t = 2$. The goal of this paper is to rule out this possibility for all symmetric composed functions of constant block-width $t > 1$.

  \subsection{Summary of Results and Comparison to Previous Protocols}
  \label{Sec:summResults}

Below, we describe our main results, and summarize in Table \ref{fig:summary} the complexity of all the known protocols for composed functions. Then, we review the main ideas used in the previous literature, and we explain how we differ from them.

\subparagraph*{Our results} In this paper, we describe the first \textit{deterministic simultaneous} protocol for symmetric composed functions of block-width $t > 1$. Our result is divided into two parts. We first give (Section \ref{Sec:equPart}) a protocol of cost $\bo{k (k+d)^{d-1} \log n}$ for $\sym \circ \sym_{\Z_d}$ when the number of players is $k \geq 4^{d-1} \log n$. In a second time (Section \ref{Sec:polyPart}), we build upon this result to give a simultaneous protocol of cost $2^{\bo{d}} \log^{2 \cdot 2^{\lceil \log d \rceil}}(n)$ for $\sym \circ \vect{\sym}_{\Z_d}$ when $k \geq 4^{2d} \log n$. Unlike the first protocol, this last result also works with different inner functions $g_1, \dots, g_n$ and it is efficient even if $k$ is super-polylogarithmic.

\begin{figure}[H]
  \begin{center}
    {\renewcommand{\arraystretch}{1.2}
    \renewcommand{\tabcolsep}{1.5mm}
    \begin{tabular}{|l|l|l|c|l|}
      \hline
                           & \begin{tabular}{c} Supported \\ functions\end{tabular} & \begin{tabular}{c} Complexity of \\ the protocol\end{tabular} & \multicolumn{1}{c|}{Simultaneous} & \begin{tabular}{c} Number of players \\ required\end{tabular} \\\hline
      Grolmusz \cite{Gro94}              & $\sym \circ \andd$              & $\bo{\log^2 n}$           & No    & $k \geq \log n$ \\ \hline
      Babai \emph{et. al.} \cite{BGKL04} & $\sym \circ \comp^c$            & $\bo{\log^{2+c} n}$       & Yes   & $k > 1 + \log n$  \\ \hline
      Ada \emph{et. al.} \cite{ACFN15}   & $\sym \circ \vect{\any}$        & $\bo{\log^3 n}$           & Yes   & $k > 1 + 2 \log n$  \\ \hline
      C. and Saks \cite{CS14}            & $\sym \circ \vect{\any}_{\Z_d}$ & $\bo{d \log n \log(dn)}$  & No    & $k > 1 + d \log(3n)$   \\ \hline
      \textbf{This work}                 & $\sym \circ \vect{\sym}_{\Z_d}$ & $2^{\bo{d}} \log^{4d}(n)$ & Yes   & $k \geq 4^{2d} \log n$  \\ \hline
    \end{tabular}}
  \end{center}
  \caption{Deterministic protocols for different families of composed functions. The top three results apply only to block-width 1 (i.e. $d=2$), whereas the last two results work for any $d$. Note that the protocol of \cite{CS14} can be made simultaneous using shared randomness between the players.}
  \label{fig:summary}
\end{figure}

{\noindent
\underline{Adjacent vertices of the $\{0,1\}^n$ hypercube.} For block-width $t=1$ and an input matrix $M \in \rn^{k \times n}$, denote $n_c$ the number of times column $c \in \rn^k$ occurs in $M$. Grolsmusz \cite{Gro94} noticed that if $c_1, \dots, c_m$ is a sequence of adjacent vertices of the $\{0,1\}^k$ hypercube (i.e. $c_{l+1}$ differs from $c_l$ by exactly one coordinate) then $n_{c_1} = \left(\sum_{l=1}^{m-1} (-1)^{l+1}(n_{c_l}+n_{c_{l+1}}) \right) + (-1)^{m+1} n_{c_m}$. Moreover, if position $i$ is the coordinate at which $c_l$ and $c_{l+1}$ differ, then the quantity $n_{c_l}+n_{c_{l+1}}$ is known by player $i$. This leads to a straightforward simultaneous protocol of cost $\bo{k \log n}$ for computing $n_{c_1}$, provided that $n_{c_m}$ is known by the referee. In his initial work, Grolsmusz \cite{Gro94} gave a non-simultaneous way to find some initial $n_{c_m}$. Ada \emph{et. al.} \cite{ACFN15} noticed later that this step can be made simultaneous using the protocol of Babai \emph{et. al.} \cite{BGKL04}, and that the idea of Grolsmusz (initially designed for $\sym \circ \andd$) easily adapts to $\sym \circ \vect{\any}$. Unfortunately, this "hypercube view" does not generalize to block-width $t > 1$: for each $i$ and $c \in (\rn^t)^k$, the number of vertices that differ from $c$ only at position $i$ is now $2^t - 1 > 1$. It is easy to see that writing a similar telescoping sum as above, in which each term would be known by a player, is no longer possible.}

~\\
{\noindent
\underline{Counting up to symmetry.} Given a $k \times n$ matrix $M$ over $\Z_d$, for all $0 \leq e_1 + \dots + e_{d-1} \leq k$ denote $y_{e_1,\dots,e_{d-1}}$ the number of columns of $M$ with exactly $e_s$ occurrences of each $s \in \Z_d \backslash \{0\}$ (we do not put $e_0$ since it is always equal to $k - (e_1 + \dots + e_{d-1})$). These numbers provide less information than the $n_c$'s defined above, but they still unable us to compute $f \circ g(M)$ for all $f \circ g \in \sym \circ \sym_{\Z_d}$. If $M$ is distributed between $k$ players in the NOF model (player $i$ does not see row $i$), a naive \textit{simultaneous} protocol is to have each player i send the number of columns $a^i_{e_1,\dots,e_{d-1}}$ which contain, \textit{from her point of view}, exactly $e_s$ occurrences of each element $s \in \{1, \dots, d-1\}$ (for all $e_1 + \dots + e_{d-1} \leq k-1$). Babai \emph{et. al.} \cite{BGKL04} analyzed this protocol in the case $d = 2$, and showed that it gives the referee enough information to recover the $y_{e_1,\dots,e_{d-1}}$'s, provided that $k > 1 + \log n$. In Section \ref{Sec:equPart}, we extend this analysis to any $d > 2$. The core of the proof, as in \cite{BGKL04}, is to define a specific equation (using the $a^i_{e_1,\dots,e_{d-1}}$'s) whose only integral solution is the $y_{e_1,\dots,e_{d-1}}$'s.}

~\\
{\noindent
\underline{The shifted basis technique.} The only protocol \cite{CS14} known prior to this work for block-width $t > 1$ is based on the following observation: given polynomial representations of the inner functions $g_j$ (over variables $x_{1,j},\dots,x_{k,j}$), each term involving strictly less than $k$ variables can be evaluated on input matrix $M$ by at least one player (in fact, by all the players that have one of the missing variables on their foreheads). The key idea of \cite{CS14} is to get rid of the remaining terms by expressing the $g_j$ in a $c$-shifted basis where all terms of degree $k$ will evaluate to $0$ on $M$ (shifting for instance  monomial $x_{1,j} \cdots x_{k,j}$ by $c = (s_1, \dots, s_k)$ means to replace it with $(x_{1,j}-s_1) \cdots (x_{k,j}-s_k)$). To this end, it would suffice to find some $c$ that shares at least one coordinate in common with each column of $M$. Provided that $k$ is large enough, \cite{CS14} showed that a randomly picked $c$ has this property with high probability. This gives rise to a simultaneous protocol for $\sym \circ \vect{\any}_{\Z_d}$ \textit{if} the players have access to a shared random string. In the deterministic setting (no shared randomness), it is not known how to make this protocol simultaneous.}

~\\
{\noindent
\underline{Different inner functions, and reducing the number of players.} The communication complexity is expected to decrease as $k$ grows up (since the overlap of information among the players increases). However, this fact is not reflected in the cost of our first protocol (Section \ref{Sec:equPart}). This issue is closely related to that of having different inner functions $g_1, \dots, g_n$. Indeed, the problem of computing $f \circ g \in \sym \circ \sym_{\Z_d}$ with $k$ players on a matrix $M \in \Z_d^{k \times n}$ can be changed into computing $f \circ (\widetilde{g_1}, \dots, \widetilde{g_n}) \in \sym \circ \vect{\sym}_{\Z_d}$ with the first $\ell < k$ players on the submatrix $\widetilde{M} \in \Z_d^{\ell \times n}$ (first $\ell$ rows of $M$), where $\widetilde{g_j} : \Z_d^{\ell} \rightarrow \{0,1\}$ is defined as $\widetilde{g_j}(u) = g(u \cdot v_j)$ and $v_j$ is the values occurring from row $\ell+1$ to $k$ in the $j$-th column of $M$ (note that the new $\widetilde{g_j}$ functions are still symmetric, but unknown to the referee). Our first protocol cannot handle different inner functions, but this issue will be solved in Section \ref{Sec:polyPart} where we describe a protocol for $\sym \circ \vect{\sym}_{\Z_d}$ based on a new use of the \textit{polynomial representations} (different than \cite{CS14}). We will show that each inner function $\widetilde{g_j}$ can be represented into a (small) basis of symmetric functions $\{m_a\}_a$ (Section \ref{Sec:polyRep}), which will allow us to split the problem of computing $f \circ (\widetilde{g_1}, \dots, \widetilde{g_n})$ on $\widetilde{M}$ into computing each $f \circ m_a \in \sym \circ \sym_{\Z_d}$ on a well-chosen matrix $\widetilde{M_a}$. This last step can be done with the initial protocol of Section \ref{Sec:equPart}.}




\section{Polynomial Representations for Symmetric Functions}
\label{Sec:polyRep}

Throughout this paper, $\Z_d$ will denote the set of integers $\{0,\dots,d-1\}$ and $\F_p$ is the finite field with $p$ elements. Furthermore, a function $f : \mathcal{X}^m \ra \mathcal{Y}$ is said to be \textit{$m$-symmetric} (or \textit{symmetric}) if it is invariant under any permutation of the input variables (i.e. for any input $(x_1,\dots,x_m)$ and permutation $\sigma \in S_m$, we have $f(x_1,\dots,x_m) = f(x_{\sigma(1)}, \dots, x_{\sigma(m)})$).

The protocol designed in Section \ref{Sec:polyPart} for composed functions $f \circ \vec{g}$ requires a concise polynomial representation of the inner functions $g_1, \dots, g_n : (\rn^t)^k \ra \rn$. Informally, we look for a field $K$ and polynomials $G_j \in K[X]$ with variables $X = (x_{u,v})_{1 \leq u \leq k, 1 \leq v \leq t}$, such that:
\begin{enumerate}
  \item[(a)] for all $x \in (\rn^t)^k$, $g_j(x) = G_j(x)$
  \item[(b)] the order of $K$ is at least $n+1$ (so that the set $\{0,\dots,n\}$ of values taken by $\sum_j g_j(x^{(j)})$ for $x^{(1)},\dots,x^{(n)} \in (\rn^t)^k$ can be embedded into $K$)
  \item[(c)] the $G_j$ polynomials can be represented in a basis of size $\bo{\poly k}$ when $t$ is constant
  \item[(d)] the values of the coefficients of the $G_j$ polynomials in this basis are less than $n^c$, for some absolute constant $c$ independent of $k$ and $t$.
\end{enumerate}

The first step towards this end is to look at the usual \textit{$\R$-multilinear representation} (also called \textit{Fourier expansion} \cite{Odo14}) of a function $g : (\rn^t)^k \ra \rn$. For each $a = (a_{u,v})_{1 \leq u \leq k, 1 \leq v \leq t}\in (\rn^t)^k$ we define the \textit{indicator polynomial} $1_{\{a\}}(x)$ to be $1_{\{a\}}(x) = \prod_{1 \leq u \leq k, 1 \leq v \leq t} (1 - a_{u,v} + (2a_{u,v}-1)x_{u,v})$. It is easy to see that it takes value $1$ when $x = a$ and value $0$ when $x \in (\rn^t)^k \setminus \{a\}$. Consequently, we have $g(x) = \sum_{a \in (\rn^t)^k} g(a) 1_{\{a\}}(x)$ for all $x \in (\rn^t)^k$. If we let $x^a$ be the monomial $\prod_{(u,v) : a_{u,v}=1} x_{u,v}$, then there exist real coefficients $\widehat{g}(a)$ such that it can be rewritten as the following multilinear polynomial
  \begin{equation}
    \label{Eq:RmultiRep}
    g(x) = \sum_{a \in (\rn^t)^k} \widehat{g}(a) x^a
  \end{equation}
Moreover, the $\widehat{g}(a)$ coefficients are given by the Möbius inversion formula
  \begin{equation}
    \label{Eq:interpCoeff}
    \widehat{g}(a) = \sum_{a' \subseteq a} (-1)^{\abs{a} - \abs{a'}} g(a')
  \end{equation}
where $\abs{a}$ is the number of $1$ in $a \in (\rn^t)^k$, and $a' \subseteq a$ means $a'_{u,v} = 0$ whenever $a_{u,v} = 0$.

Polynomial (\ref{Eq:RmultiRep}) is called the \textit{$\R$-multilinear representation} of function $g$. It satisfies requirements (a) and (b) above, but not requirement (c). Indeed, these polynomials are expressed in the basis of monomials $\{x^a\}_{a \in (\rn^t)^k}$ which has size $2^{t \cdot k}$.

In order to reduce the size of the basis, we restrict ourselves to the $k$-symmetric functions $g : (\rn^t)^k \ra \rn$ (as will be the case in Section \ref{Sec:polyPart}). This condition leads to the following equalities between coefficients.

\begin{lemma}
  \label{Lem:permCoeff}
  For any $a = (a_1,\dots,a_k) \in (\rn^t)^k$ and any permutation $\sigma \in S_k$, if $g : (\rn^t)^k \ra \rn$ is a $k$-symmetric function then the coefficients $\widehat{g}(a)$ and $\widehat{g}(\sigma(a))$ in the $\R$-multilinear representation of $g$ are equal (where $\sigma(a) = (a_{\sigma(1)},\dots,a_{\sigma(k)})$).
\end{lemma}

\begin{proof}
  The proof is direct from Equation (\ref{Eq:interpCoeff}).
\end{proof}

This lemma motivates the definition of the following polynomials, that will be used to obtain a basis for the $k$-symmetric functions over $(\rn^t)^k$.

\begin{definition}
  Given $a \in (\rn^t)^k$, the \emph{monomial $k$-symmetric polynomial} $m_a(x)$ over variables $(x_{u,v})_{1 \leq u \leq k, 1 \leq v \leq t}$ is defined to be the sum of all the \textit{distinct} monomials $x^{\sigma(a)}$ where $\sigma \in S_k$ ranges over all the permutations.
\end{definition}

\begin{example}
  If $(t,k)=(2,3)$ and $a=((1,1),(0,1),(0,1))$ then $m_a(x) = x_{1,1}x_{1,2}x_{2,2}x_{3,2} + x_{1,2}x_{2,1}x_{2,2}x_{3,2} + x_{1,2}x_{2,2}x_{3,1}x_{3,2}$.
\end{example}

According to Lemma \ref{Lem:permCoeff}, any $k$-symmetric function $g : (\rn^t)^k \ra \rn$ can be expressed as a linear combination of monomial $k$-symmetric polynomials. From this observation, we can derive a basis for the $k$-symmetric functions by taking all the \emph{distinct} monomial $k$-symmetric polynomials. We specify a subset of elements $a \in (\rn^t)^k$ that corresponds to this basis.

\begin{definition}
  \label{Def:ordered}
  We define a tuple $a = (a_1, \dots, a_k) \in (\rn^t)^k$ to be \textit{sorted}, if $\abs{a_u} \leq \abs{a_{u'}}$ for all $1 \leq u \leq u' \leq k$, and $a_u \leq_{lex} a_{u'}$ whenever $\abs{a_u} = \abs{a_{u'}}$ (where $\abs{a_u}$ is the Hamming weight of $a_u$, and $\leq_{lex}$ is the lexicographic order over $\rn^t$). The set of all the sorted tuples over $(\rn^t)^k$ is denoted $\sor(t,k)$.
\end{definition}

\begin{lemma}
  \label{Lem:symBasis}
  The set $\left\{m_a(x) : a \in \sor(t,k)\right\}$ is a basis for the $k$-symmetric functions $g : (\rn^t)^k \ra \rn$. Moreover, it has size $\binom{k + 2^t - 1}{2^t-1}$.
\end{lemma}

\begin{proof}
  It is straightforward to see that all the possible monomial $k$-symmetric polynomials belong to $\left\{m_a(x) : a \in \sor(t,k)\right\}$, and that no two elements in this set have a monomial in common. Thus, it is a basis for the $k$-symmetric functions.

  Consider the total order $\prec$ over $\rn^t$ defined as $a_u \prec a_{u'}$ if and only if $\abs{a_u} \leq \abs{a_{u'}}$, or $\abs{a_u} = \abs{a_{u'}}$ and $a_u \leq_{lex} a_{u'}$. Each $a \in \sor(t,k)$ can be seen as a (distinct) non-decreasing sequence of length $k$ from the totally ordered set $(\rn^t,\prec)$ of size $2^t$. The total number of such sequences is known to be $\binom{k + 2^t - 1}{2^t-1}$.
\end{proof}

Finally, given a parameter $n$, we want the coefficients of the $k$-symmetric functions in the chosen basis to be less than $n^c$ for some constant $c$ independent of $k$ and $t$ (requirement (d)). To this end, it suffices to reformulate the previous results over a field $\F_p$, for some prime $p \in (n,2n)$. We obtain the following polynomial representation for $k$-symmetric functions:

\begin{proposition}
  \label{Prop:polyRep}
  Any $k$-symmetric function $g : (\rn^t)^k \ra \rn$ can be written as
    \[g(x) = \sum\limits_{a \in \sor(t,k)} c_a(g) \cdot m_a(x) \mod p\]
  where $p \in (n,2n)$ is prime, $c_a(g) \in \F_p$ and $m_a$ is the monomial $k$-symmetric polynomial corresponding to the sorted tuple $a$. Moreover, $\sor(t,k)$ has size $\binom{k + 2^t - 1}{2^t-1}$.
\end{proposition}


\section[Simultaneous Protocol for SymSymZd]{Simultaneous Protocol for $\sym \circ \protect\vect{\sym}_{\Z_d}$}
\label{Sec:mainSec}

We now describe in detail our simultaneous protocol for symmetric composed functions. The result is divided into two parts. We first give in Section \ref{Sec:equPart} a protocol of cost $\bo{k(k+d)^{d-1} \log n}$ for $\sym \circ \sym_{\Z_d}$ when $k \geq 4^{d-1} \log n$. This is a generalization of the idea of \cite{BGKL04}, which was based on solving a particular equation. We build upon this result in Section \ref{Sec:polyPart} to give an efficient protocol of cost $\bo{\log^{4d}(n)}$ for $\sym \circ \vect{\sym}_{\Z_d}$ when $k \geq 4^{2d} \log n$ and $d$ is constant. This last result uses the protocol of Theorem \ref{Thm:symSym} as a subroutine, and the polynomial representations described in Section \ref{Sec:polyRep}.


  \subsection{The Equation Solving part}
  \label{Sec:equPart}

We extend the protocol for $\sym \circ \sym_{\Z_2}$ from \cite{BGKL04} to any $d > 1$. It applies to all functions in $\sym \circ \sym_{\Z_d}$ as long as $k \geq 4^{d-1} \log n$, but it is not efficient if $d$ is nonconstant or if the number $k$ of players is super-polylogarithmic (we will remove this last condition in the next section). For convenience in the proof, we state the result over $\Z_{d+1}$ instead of $\Z_d$:


\begin{theorem}
  \label{Thm:symSym}
  Let $M$ be a $k \times n$ matrix over $\Z_{d+1}$, where $n \geq 2$ and $d \geq 1$. For $0 \leq e_1 + \dots + e_d \leq k$, denote $y_{e_1,\dots,e_d}$ the number of columns of $M$ with exactly $e_s$ occurrences of each $s \in \Z_{d+1} \backslash \{0\}$. For each $i = 1, \dots, k$, let player $i$ see all of $M$ except row $i$. If $k \geq 4^d \log n$ then there exists a deterministic simultaneous NOF protocol of cost $k \binom{k+d}{d} \lceil\log n\rceil$, at the end of which the referee knows all the $y_{e_1,\dots,e_d}$'s.
\end{theorem}

\begin{proof}
  The communication part of the protocol is pretty simple: each player $i$ sends to the referee the number of columns $a^i_{e_1,\dots,e_d}$ which contain, \textit{from her point of view} (i.e. without taking row $i$ into account), exactly $e_s$ occurrences of each element $s \in \{1, \dots, d\}$ (for all $e_1 + \dots + e_d \leq k-1$).

  The referee computes then $b_{e_1,\dots,e_d} = \sum_{i=1}^k a^i_{e_1,\dots,e_d}$ (for all $e_1 + \dots + e_d \leq k-1$). The important thing to note is that these numbers must verify the following equalities:
  \begin{equation}
    \label{Eq:simultEq}
    \left\{
      \def\arraystretch{1.5}
      \begin{array}{l l}
        (k-(e_1 + \dots + e_d)) y_{e_1,\dots,e_d} + \sum\limits_{s=1}^d (e_s + 1) y_{e_1,\dots,e_{s-1},e_s+1,e_{s+1},\dots,e_d} = b_{e_1,\dots,e_d} \\
        0 \leq e_1 + \dots + e_d \leq k-1
      \end{array}
    \right.
  \end{equation}

  To see why it is true, consider a column $C$ of $M$ that contributes to a given $b_{e_1,\dots,e_d}$. Either $C$ contains exactly $e_s$ occurrences of each element $s \in \{1, \dots, d\}$, or there is one $s' \in \{1, \dots, d\}$ that occurs $e_{s'} + 1$ times in $C$ (the other $s$ having exactly $e_s$ occurrences in $C$). In the first case, $C$ contributes to $y_{e_1,\dots,e_d}$ and to the quantity $a_i(e_1,\dots,e_d)$ of each player $i$ having a $0$ entry of $C$ on her forehead (there are $k-(e_1 + \dots + e_d)$ such players). In the second case, $C$ contributes to $y_{e_1,\dots,e_{s'-1},e_{s'}+1,e_{s'+1},\dots,e_d}$ and to the quantity $a_i(e_1,\dots,e_d)$ of each player $i$ having a $s'$ entry of $C$ on her forehead (there are $e_{s'} + 1$ such players). Thus, the total contribution for $b_{e_1,\dots,e_d}$ is $(k-(e_1 + \dots + e_d)) y_{e_1,\dots,e_d} + \sum_{s'=1}^d (e_{s'} + 1) y_{e_1,\dots,e_{s'-1},e_{s'}+1,e_{s'+1},\dots,e_d}$.

  Equalities (\ref{Eq:simultEq}) can be seen as a system of equations whose unknowns are the $y_{e_1,\dots,e_d}$'s. Since the referee is not computationally restricted she can enumerate all the integral solutions, but she does not know which one corresponds to matrix $M$. The key lemma is to show that Equations (\ref{Eq:simultEq}), under mild constraints
  \begin{equation}
    \label{Eq:simulCons}
    y_{e_1,\dots,e_d} \geq 0,\ 0 \leq e_1 + \dots + e_d \leq k \quad \textit{and} \quad \sum\limits_{e_1 + \dots + e_d \leq k} y_{e_1,\dots,e_d} \leq n
  \end{equation}
 have at most one integral solution when $k \geq 4^d \log n$. We prove it by induction on $d$ (the base case $d = 1$ corresponds to the work of \cite{BGKL04}, the induction step is more involved and is given in Appendix \ref{App:babaiGen}). Consequently, the referee is able to know unambiguously the correct $y_{e_1,\dots,e_d}$'s that correspond to $M$.

  This protocol is clearly simultaneous since the players do not need to talk to each other. Each of the $k$ players sends $\binom{k+d}{d}$ numbers $a_i(e_1,\dots,e_d) \leq n$. Thus the total communication cost is at most $k \binom{k+d}{d} \lceil\log n\rceil$.
\end {proof}

\begin{corollary}
  \label{Cor:symSym}
  Let $n \geq 2$, $d \geq 2$ and suppose $k \geq 4^{d-1} \log n$. There is a deterministic simultaneous NOF protocol of cost $k \binom{k+d-1}{d-1} \lceil\log n\rceil$, at the end of which the referee can compute \emph{all} composed functions $f \circ g \in \sym \circ \sym_{\Z_d}$ of her choice.
\end{corollary}


This result can also be adapted to the case of $k < 4^d \log n$ players by splitting the initial matrix into sufficiently many parts. Previously, Ada \emph{et. al.} \cite{ACFN15} also generalized their work to any number $k$ of players, by giving a protocol of cost $\bo{n/2^k \cdot \log n + k \log n}$ for $\sym \circ \vect{\any}$. However, it was not simultaneous and it does not apply to $t > 1$.

\begin{proposition}
  \label{Prop:symSymBelow}
  Let $M$ be a $k \times n$ matrix over $\Z_{d+1}$, where $n \geq 2$ and $d \geq 1$. For $0 \leq e_1 + \dots + e_d \leq k$, denote $y_{e_1,\dots,e_d}$ the number of columns of $M$ with exactly $e_s$ occurrences of each $s \in \Z_{d+1} \backslash \{0\}$. For each $i = 1, \dots, k$, let player $i$ see all of $M$ except row $i$. If $4^d \leq k < 4^d \log n$ then there exists a deterministic simultaneous NOF protocol of cost at most $\bo{\frac{n}{2^{k/4^d}} \cdot (k+d)^{d+2}}$, at the end of which the referee knows all the $y_{e_1,\dots,e_d}$'s.
\end{proposition}

\begin{proof}
  We split $M$ into $\left\lceil \frac{n}{\left\lfloor 2^{k/4^d} \right\rfloor} \right\rceil$ matrices, each of size $k \times \left\lfloor 2^{k/4^d} \right\rfloor$ (except one matrix that can have less columns). These matrices have few enough columns to apply (separately) the protocol of Theorem \ref{Thm:symSym} on them. The $y_{e_1,\dots,e_d}$'s for the original matrix $M$ are computed by recombining all the obtained results. The total cost is $\bo{\frac{n}{2^{k/4^d}} \cdot k \binom{k+d}{d} \log\left(2^{k/4^d}\right)}$.
\end{proof}


  \subsection{The Polynomial Representation part}
  \label{Sec:polyPart}

Using the polynomial representation of Proposition \ref{Prop:polyRep}, we give a protocol that improves upon Corollary \ref{Cor:symSym} in two ways: it is still efficient when $k$ is super-polylogarithmic, and the inner functions $g_1, \dots, g_n$ can be different (i.e. it applies to $\sym \circ \vect{\sym}_{\Z_d}$ instead of $\sym \circ \sym_{\Z_d}$).



\begin{theorem}
  \label{Thm:symVecSym}
  Let $n \geq 2$, $d \geq 2$ and suppose $k \geq 4^{2^{\lceil \log d \rceil}} \log n$. For any composed function $f \circ \vec{g} \in \sym \circ \vect{\sym}_{\Z_d}$ there exists a deterministic simultaneous NOF protocol that computes it with cost $4^{2^{\lceil \log d \rceil + 2}} \log^{2 \cdot 2^{\lceil \log d \rceil}}(n)$.
\end{theorem}

\begin{proof}
  Let $\vec{g} = (g_1,\dots,g_n)$. In order to use the polynomial representation of Section \ref{Sec:polyRep}, we change the range of the $g_j$ functions as $g_j : (\rn^t)^k \ra \rn$, where $t = \lceil \log d \rceil$. This requires to encode each number $x \in \Z_d$ as an element $\bar{x} \in \rn^t$. If $d$ is not a power of two then some $y \in \rn^t$ will not correspond to any $x \in \Z_d$. We extend each $g_j$ as the zero function on inputs that contain such numbers (note that the functions are still $k$-symmetric).

  The input is now a $k \times (t \cdot n)$ boolean matrix $M$. Each function $g_j$ acts on the $j^{th}$ block of $M$, which will be denoted $B_j \in (\rn^t)^k$. Let $\ell = 4^{2^t} \log n$, so that only the first $\ell$ players are going to speak. For each block $B_j$, if we let $v_j \in (\rn^t)^{(k-\ell)}$ be the sub-block occurring from row $\ell+1$ to $k$, then $g_j : (\rn^t)^k \rightarrow \{0,1\}$ induces a new function $\widetilde{g_j} : (\rn^t)^{\ell} \rightarrow \{0,1\}$ such that $\widetilde{g_j}(u) = g_j(u \cdot v_j)$. Moreover, $\widetilde{g_j}$ is still a symmetric function. Thus, our task reduces to find an efficient simultaneous protocol for $f \circ (\widetilde{g_1}, \dots, \widetilde{g_n})$ with $\ell = 4^{2^t} \log n$ players. We denote $\widetilde{M}$ the $\ell \times (t \cdot n)$ submatrix of $M$ on which we now work, and $\widetilde{B_j} \in (\rn^t)^{\ell}$ is the sub-block of $B_j$ corresponding to $\widetilde{M}$.

  We cannot apply directly the protocol of Theorem \ref{Thm:symSym}, since it only works for equal inner functions $\widetilde{g_1} = \dots = \widetilde{g_n}$. Instead, we use first Proposition \ref{Prop:polyRep} on the $\widetilde{g_j}$ functions: for each $j \in \{1, \dots, n\}$ there exist coefficients $(c_a(\widetilde{g_j}))_{a \in \sor(t,\ell)}$ over $\F_p$ such that $\widetilde{g_j}(x) = \sum_{a \in \sor(t,\ell)} c_a(\widetilde{g_j}) \cdot m_a(x) \mod p$ where $p \in (n,2n)$, $m_a$ is the monomial $k$-symmetric polynomial corresponding to the sorted tuple $a$ and $\abs{\sor(t,\ell)} = \binom{\ell + 2^t - 1}{2^t-1}$. The coefficients $c_a(\widetilde{g_j})$ are known by the first $\ell$ players, but not by the referee (since they depend on rows $\ell+1$ to $k$ of $M$).

  For each $a \in \sor(t,\ell)$, the players build a new matrix $\widetilde{M_a}$ of size $\ell \times (c_a(\widetilde{g_1}) + \dots + c_a(\widetilde{g_n}))$ where block $\widetilde{B_j}$ from $\widetilde{M}$ is copied $c_a(\widetilde{g_j}) \in [0,2n)$ times. Note that $\widetilde{M_a}$ has at most $2n^2$ blocks, and there are enough players $\ell = 4^{2^t} \log n$ for applying (the boolean input version of) the simultaneous protocol of Theorem \ref{Thm:symSym}. It allows the referee to know the number of blocks of $\widetilde{M_a}$ which are equal ---up to row permutation--- to any $\widetilde{B} \in (\rn^t)^{\ell}$ . This information is sufficient to compute $\sum_{j=1}^n c_a(\widetilde{g_j}) \cdot m_a(\widetilde{B_j})$ since the $m_a$ functions are $k$-symmetric.

  Finally, the referee sums these quantities modulo $p$ over all $a$. It gives her $\sum_{a \in \sor(t,\ell)} \\ \sum_{j=1}^n c_a(\widetilde{g_j}) \cdot m_a(\widetilde{B_j}) \mod p = \sum_{j=1}^n \widetilde{g_j}(\widetilde{B_j}) \mod p$. Since $\sum_{j=1}^n \widetilde{g_j}(\widetilde{B_j}) \leq n$ and $p > n$, it equals $\sum_{j=1}^n \widetilde{g_j}(\widetilde{B_j}) = \sum_{j=1}^n g_j(B_j)$. Knowing this, the referee can compute $f \circ (g_1, \dots, g_n)(M)$ since $f$ is symmetric.

  Regarding the cost of the protocol, we applied $\abs{\sor(t,\ell)} = \binom{\ell + 2^t - 1}{2^t-1}$ times the protocol of Theorem \ref{Thm:symSym}, with $\ell$ players and inputs of size at most $2 n^2$. Thus the total cost is at most $\binom{\ell + 2^t - 1}{2^t-1} \cdot \ell \binom{\ell+2^t-1}{2^t-1} \lceil\log 2n^2\rceil \leq \ell (\ell + 2^t)^{2^{t+1}-2} \log n$. Since $\ell = 4^{2^t} \log n$ and $t = \left\lceil \log d \right\rceil$, this is less than $4^{2^t + (2^t+1)(2^{t+1}-2)} \log^{2^{t+1}}(n) \leq 4^{2^{\lceil \log d \rceil + 2}} \log^{2 \cdot 2^{\lceil \log d \rceil}}(n)$.
\end{proof}

\section{Conclusion and Open Problems}

One of the main open problems in communication complexity remains to find a function which is hard to compute for $k \geq \log n$ players in the simultaneous Number On the Forehead model. We discarded this possibility for the composed functions in $\sym \circ \vect{\sym}_{\Z_d}$ (for constant $d$) by giving the first efficient \textit{deterministic simultaneous} protocol for composed functions of block-width $t > 1$. In the non-simultaneous setting, the best result so far applies to $\sym \circ \vect{\any}_{\Z_d}$ and $d = \bo{\polylog n}$ \cite{CS14}. Extending these protocols to larger $d$, bigger families of composed functions or to the simultaneous setting (for \cite{CS14}) would give a better insight on composed functions. Indeed, it is conjectured that the $\log n$ barrier can be broken by such functions for large $d$, two of the candidates being $\maj \circ \maj_t$ and $\maj \circ \thr^s_t$.

Note that both the Equation Solving and the Polynomial Representation parts of our protocol are bottleneck for handling non-constant $d$ in our result. It could be interesting to restrict to smaller families than symmetric functions (or to choose specific inner or outer functions, such as threshold functions), or to find other relevant equations that could be solved by the referee with fewer information than in our protocol.

Apart from composed functions, there are a few other candidates for breaking the $\log n$ barrier. Some of them are matrix related problems, such as deciding the top-left entry of the multiplication of $k$ matrices in $\F^{n \times n}_2$ (an $\Omega(n/2^k)$ lower bound has been obtained by Raz \cite{Raz00}). More recently, Gowers and Viola \cite{GV15} studied the interleaved group products, where each player receives a tuple $(x_{i,1},\dots,x_{i,n})$ in $G = SL(2,q)$, with the promise that $\prod_{i=1}^n x_{1,i} \cdots x_{k,i} = g \textit{ or } h$. Finding which is the case has cost $\Omega(n \log \abs{G})$ when $k = 2$, and it is conjectured to remain hard for larger $k$.


\section*{Acknowledgements}

This work was initiated during a visit to Carnegie Mellon University. The author is very grateful to Anil Ada, who introduced him to the $\log n$ barrier problem and the $\maj \circ \maj_t$ conjecture for composed functions. He also thanks him for helpful discussions on this subject, as well as the anonymous referees for their valuable comments and suggestions which helped to improve this paper.


\bibliographystyle{alpha}
\bibliography{ComposedFunctions}

\begin{thebibliography}{BCWdW01}

\bibitem[ACFN15]{ACFN15}
A.~Ada, A.~Chattopadhyay, O.~Fawzi, and P.~Nguyen.
\newblock The {NOF} multiparty communication complexity of composed functions.
\newblock {\em Computational Complexity}, 24(3):645--694, 2015.

\bibitem[Amb96]{Amb96}
A.~Ambainis.
\newblock Communication complexity in a 3-computer model.
\newblock {\em Algorithmica}, 16(3):298--301, 1996.

\bibitem[BCWdW01]{BCWW01}
H.~Buhrman, R.~Cleve, J.~Watrous, and R.~de~Wolf.
\newblock Quantum fingerprinting.
\newblock {\em Phys. Rev. Lett.}, 87:167902, 2001.

\bibitem[BGK15]{BGK15}
R.~C. Bottesch, D.~Gavinsky, and H.~Klauck.
\newblock Equality, revisited.
\newblock {\em CoRR}, abs/1511.01211, 2015.

\bibitem[BGKL04]{BGKL04}
L.~Babai, A.~G\'{a}l, P.~G. Kimmel, and S.~V. Lokam.
\newblock Communication complexity of simultaneous messages.
\newblock {\em SIAM J. Comput.}, 33(1):137--166, 2004.

\bibitem[BH12]{BH12}
P.~Beame and T.~Huynh.
\newblock Multiparty communication complexity and threshold circuit size of
  {AC0}.
\newblock {\em SIAM Journal on Computing}, 41(3):484--518, 2012.

\bibitem[BK97]{BK97}
L.~Babai and P.~G. Kimmel.
\newblock Randomized simultaneous messages: solution of a problem of {Yao} in
  communication complexity.
\newblock In {\em Proceedings of Computational Complexity. Twelfth Annual IEEE
  Conference}, pages 239--246, 1997.

\bibitem[BKL95]{BKL95}
L.~Babai, P.~G. Kimmel, and S.~V. Lokam.
\newblock Simultaneous messages vs. communication.
\newblock In {\em 12th Annual Symposium on Theoretical Aspects of Computer
  Science (STACS)}, pages 361--372. Springer, 1995.

\bibitem[BNS92]{BNS92}
L.~Babai, N.~Nisan, and M.~Szegedy.
\newblock Multiparty protocols, pseudorandom generators for logspace, and
  time-space trade-offs.
\newblock {\em J. Comput. Syst. Sci.}, 45(2):204--232, 1992.

\bibitem[BPS07]{BPS07}
P.~Beame, T.~Pitassi, and N.~Segerlind.
\newblock Lower bounds for {Lov\'{a}sz--Schrijver} systems and beyond follow
  from multiparty communication complexity.
\newblock {\em {SIAM} J. Comput.}, 37(3):845--869, 2007.

\bibitem[BPSW05]{BPSW05}
P.~Beame, T.~Pitassi, N.~Segerlind, and A.~Wigderson.
\newblock A direct sum theorem for corruption and the multiparty {NOF}
  communication complexity of set disjointness.
\newblock In {\em Proceedings of the 20th Annual IEEE Conference on
  Computational Complexity}, CCC '05, pages 52--66. IEEE Computer Society,
  2005.

\bibitem[BPSW06]{BPS06}
P.~Beame, T.~Pitassi, N.~Segerlind, and A.~Wigderson.
\newblock A strong direct product theorem for corruption and the multiparty
  communication complexity of disjointness.
\newblock {\em Comput. Complex.}, 15(4):391--432, 2006.

\bibitem[BT94]{BT94}
R.~Beigel and J.~Tarui.
\newblock On {ACC}.
\newblock {\em Computational Complexity}, 4(4):350--366, 1994.

\bibitem[BYJKS02]{BJKS02}
Z.~Bar-Yossef, T.~S. Jayram, R.~Kumar, and D.~Sivakumar.
\newblock Information theory methods in communication complexity.
\newblock In {\em Proceedings 17th IEEE Annual Conference on Computational
  Complexity}, pages 72--81, 2002.

\bibitem[CA08]{CA08}
A.~Chattopadhyay and A.~Ada.
\newblock Multiparty communication complexity of disjointness.
\newblock {\em arXiv preprint arXiv:0801.3624}, 2008.

\bibitem[CFL83]{CFL83}
A.~K. Chandra, M.~L. Furst, and R.~J. Lipton.
\newblock Multi-party protocols.
\newblock In {\em Proceedings of the Fifteenth Annual ACM Symposium on Theory
  of Computing}, STOC '83, pages 94--99. ACM, 1983.

\bibitem[CS14]{CS14}
A.~Chattopadhyay and M.~E. Saks.
\newblock The power of super-logarithmic number of players.
\newblock In {\em Approximation, Randomization, and Combinatorial Optimization.
  Algorithms and Techniques (APPROX/RANDOM)}, 2014.

\bibitem[CT93]{CT93}
F.~R.~K. Chung and P.~Tetali.
\newblock Communication complexity and quasi randomness.
\newblock {\em SIAMJDiscreteMath}, 6(1):110--123, 1993.

\bibitem[GRd08]{GRW08}
D.~Gavinsky, O.~Regev, and R.~{de Wolf}.
\newblock Simultaneous communication protocols with quantum and classical
  messages.
\newblock {\em Chicago Journal of Theoretical Computer Science}, 7, 2008.

\bibitem[Gro94]{Gro94}
V.~Grolmusz.
\newblock The {BNS} lower bound for multi-party protocols is nearly optimal.
\newblock {\em Information and Computation}, 112:51--54, 1994.

\bibitem[GV15]{GV15}
T.~Gowers and E.~Viola.
\newblock The communication complexity of interleaved group products.
\newblock In {\em Proceedings of the Forty-Seventh Annual ACM on Symposium on
  Theory of Computing}, STOC '15, pages 351--360. ACM, 2015.

\bibitem[HG91]{HG91}
J.~H{\aa}stad and M.~Goldmann.
\newblock On the power of small-depth threshold circuits.
\newblock {\em Computational Complexity}, 1(2):113--129, 1991.

\bibitem[LS09]{LS09}
T.~Lee and A.~Shraibman.
\newblock Disjointness is hard in the multiparty number-on-the-forehead model.
\newblock {\em Computational Complexity}, 18(2):309--336, 2009.

\bibitem[NS96]{NS96}
I.~Newman and M.~Szegedy.
\newblock Public vs. private coin flips in one round communication games
  (extended abstract).
\newblock In {\em Proceedings of the Twenty-eighth Annual ACM Symposium on
  Theory of Computing}, STOC '96, pages 561--570. ACM, 1996.

\bibitem[NW93]{Nis93}
N.~Nisan and A.~Wigderson.
\newblock Rounds in communication complexity revisited.
\newblock {\em SIAM J. Comput.}, 22(1):211--219, 1993.

\bibitem[O'D14]{Odo14}
R.~O'Donnell.
\newblock {\em Analysis of Boolean Functions}.
\newblock Cambridge University Press, 2014.

\bibitem[PRS97]{PRS97}
P.~Pudl\'{a}k, V.~R\"{o}dl, and J.~Sgall.
\newblock Boolean circuits, tensor ranks, and communication complexity.
\newblock {\em SIAM J. Comput.}, 26(3):605--633, 1997.

\bibitem[PS17]{PS17}
V.~V. Podolskii and A.~A. Sherstov.
\newblock Inner product and set disjointness: Beyond logarithmically many
  parties.
\newblock {\em CoRR}, abs/1711.10661, 2017.

\bibitem[Raz00]{Raz00}
R.~Raz.
\newblock The {BNS-Chung} criterion for multi-party communication complexity.
\newblock {\em Computational Complexity}, 9:2000, 2000.

\bibitem[RY15]{RY15}
A.~Rao and A.~Yehudayoff.
\newblock Simplified lower bounds on the multiparty communication complexity of
  disjointness.
\newblock In {\em Proceedings of the 30th Conference on Computational
  Complexity}, CCC '15, pages 88--101, Germany, 2015. Schloss
  Dagstuhl--Leibniz-Zentrum fuer Informatik.

\bibitem[She11]{She11}
A.~A. Sherstov.
\newblock The pattern matrix method.
\newblock {\em SIAM Journal on Computing}, 40(6):1969--2000, 2011.

\bibitem[She14]{She14}
A.~A. Sherstov.
\newblock Communication lower bounds using directional derivatives.
\newblock {\em J. ACM}, 61(6):34:1--34:71, 2014.

\bibitem[She16]{She16}
A.~A. Sherstov.
\newblock The multiparty communication complexity of set disjointness.
\newblock {\em SIAM Journal on Computing}, 45(4):1450--1489, 2016.

\bibitem[Tes03]{Tes03}
P.~Tesson.
\newblock {\em Computational Complexity Questions Related to Finite Monoids and
  Semigroups}.
\newblock PhD thesis, Montreal, Canada, 2003.

\bibitem[Wil14]{Wil14}
R.~Williams.
\newblock Nonuniform {ACC} circuit lower bounds.
\newblock {\em J. ACM}, 61(1):2:1--2:32, 2014.

\bibitem[Yao79]{Yao79}
A.~Yao.
\newblock Some complexity questions related to distributive computing.
\newblock In {\em Proceedings of the Eleventh Annual ACM Symposium on Theory of
  Computing}, STOC '79, pages 209--213. ACM, 1979.

\bibitem[Yao90]{Yao90}
A.~Yao.
\newblock On {ACC} and threshold circuits.
\newblock In {\em Proceedings 31st Annual Symposium on Foundations of Computer
  Science}, pages 619--627 vol.2, 1990.

\end{thebibliography}

\appendix


  \section{Lemma for the Equation Solving part}
  \label{App:babaiGen}

In this section, we prove the following lemma:

\begin{lemma}
  \label{Lem:babaiGen}
  Let $n \geq 2$, $d \geq 1$ and $k \geq 4^d \log n$. Let $(b_{e_1,\dots,e_d})_{0 \leq e_1 + \dots + e_d \leq k-1}$ be integers. Consider the following system of equations:
    \begin{equation}
      \label{Eq:babaiGenEq}
      \left\{
        \def\arraystretch{1.5}
        \begin{array}{l l}
          (k-(e_1 + \dots + e_d)) y_{e_1,\dots,e_d} + \sum\limits_{s=1}^d (e_s + 1) y_{e_1,\dots,e_{s-1},e_s+1,e_{s+1},\dots,e_d} = b_{e_1,\dots,e_d} \\
          0 \leq e_1 + \dots + e_d \leq k-1
        \end{array}
      \right.
    \end{equation}

  Assume further that
    \begin{equation}
      \label{Eq:babaiGenCons}
      y_{e_1,\dots,e_d} \geq 0,\ 0 \leq e_1 + \dots + e_d \leq k \quad \textit{and} \quad \sum\limits_{e_1 + \dots + e_d \leq k} y_{e_1,\dots,e_d} \leq n
    \end{equation}

  Then, under constraints (\ref{Eq:babaiGenCons}), the system of equations (\ref{Eq:babaiGenEq}) has at most one integral solution.
\end{lemma}

In fact, we are going to show a stronger result:

\begin{lemma}
  \label{Lem:babaiGenBis}
  Let $n \geq 2$, $d \geq 1$ and $k > 4^d \log n - d$. Let $(b_{e_1,\dots,e_d})_{0 \leq e_1 + \dots + e_d \leq k-1}$ be integers. Consider the following system of equations:
    \begin{equation}
      \label{Eq:babaiGenEqBis}
      \left\{
        \def\arraystretch{1.5}
        \begin{array}{l l}
          (k-(e_1 + \dots + e_d)) z_{e_1,\dots,e_d} + \sum\limits_{s=1}^d (e_s + 1) z_{e_1,\dots,e_{s-1},e_s+1,e_{s+1},\dots,e_d} = 0 \\
          0 \leq e_1 + \dots + e_d \leq k-1
        \end{array}
      \right.
    \end{equation}

  Assume further that
    \begin{equation}
      \label{Eq:babaiGenConsBis}
      \sum\limits_{e_1 + \dots + e_d \leq k} |z_{e_1,\dots,e_d}| \leq 2n
    \end{equation}

  Then, under constraints (\ref{Eq:babaiGenConsBis}), the system of equations (\ref{Eq:babaiGenEqBis}) cannot have a non-zero integral solution.
\end{lemma}

\begin{proof}[Proof that Lemma \ref{Lem:babaiGenBis} implies Lemma \ref{Lem:babaiGen}]
  Assume by contradiction that Equations (\ref{Eq:babaiGenEq}) under Constraints (\ref{Eq:babaiGenCons}) have two different integral solutions $y = (y_{e_1,\dots,e_d})_{0 \leq e_1 + \dots + e_d \leq k}$ and $y' = (y'_{e_1,\dots,e_d})_{0 \leq e_1 + \dots + e_d \leq k}$ for $k \geq 4^d \log n$. Define $z_{e_1,\dots,e_d} = y_{e_1,\dots,e_d} - y'_{e_1,\dots,e_d}$. It is easy to see that it must verify (\ref{Eq:babaiGenEqBis}), and since $y \neq y'$ there is at least one $z_{e_1,\dots,e_d} \neq 0$. Finally, since $z_{e_1,\dots,e_d} = |y_{e_1,\dots,e_d} - y'_{e_1,\dots,e_d}| \leq y_{e_1,\dots,e_d} + y'_{e_1,\dots,e_d}$, we have
    \[\sum\limits_{e_1 + \dots + e_d \leq k} |z_{e_1,\dots,e_d}| \leq \sum\limits_{e_1 + \dots + e_d \leq k} (y_{e_1,\dots,e_d} + y'_{e_1,\dots,e_d}) \leq 2n\]
\end{proof}

\begin{proof}[Proof of Lemma \ref{Lem:babaiGenBis}]
  We prove the result by induction on $d$. The base case has already been established in \cite{BGKL04}, we recall it for completeness.

  \textbf{Base case (d = 1).}
  We denote $(z_i)_{0 \leq i \leq k}$ the variables. Equations (\ref{Eq:babaiGenEqBis}) under Constraints (\ref{Eq:babaiGenConsBis}) become
    \[
    \left\{
      \def\arraystretch{1.5}
      \begin{array}{l l}
        (k-i)z_i + (i+1)z_{i+1} = 0,\ i = 0,1,\dots,k-1 \\
        \sum_{i=0}^k |z_i| \leq 2n
      \end{array}
    \right.
  \]
  Thus, $z_1 = -k z_0 = -\binom{k}{1} z_0$, $z_2 = -\frac{k-1}{2} z_1 = \binom{k}{2} z_0$, and more generally $z_i = (-1)^i \binom{k}{i} z_0$. Consequently, if $(z_i)_{0 \leq i \leq k}$ is a nonzero integral solution, then $z_0 \neq 0$ and $\abs{z_i} = \binom{k}{i} \abs{z_0} \geq \binom{k}{i}$ for all $i$. We obtain a contradiction: $2n \geq \sum_{i=0}^k \abs{z_i} \geq \sum_{i=0}^k \binom{k}{i} = 2^k > 2^{4 \log n - 1} > 2n$. Thus, Lemma \ref{Lem:babaiGenBis} is true for $d = 1$.

  \textbf{Induction step.}
  Assuming that Lemma \ref{Lem:babaiGenBis} is true for $d-1$, we prove that it is also the case for $d \geq 2$. Suppose by contradiction that Equations (\ref{Eq:babaiGenEqBis}) under Constraints (\ref{Eq:babaiGenConsBis}) have a non-zero integral solution $z = (z_{e_1,\dots,e_d})_{0 \leq e_1 + \dots + e_d \leq k}$ for $k > 4^d \log n - d$. As in the proof of the base case, we want to show $\sum_{e_1 + \dots + e_d \leq k} |z_{e_1,\dots,e_d}| > 2n$, which would be a contradiction.

  To this end, we are going to focus for each $0 \leq i \leq k$ on the largest element of $\{|z_{e_1,\dots,e_d}| : e_1 + \dots + e_d = i\}$. We define
    \[\qquad \qquad Z_i = \max\limits_{e_1 + \dots + e_d = i} |z_{e_1,\dots,e_d}| \quad \text{ and } \quad \kp = \min \{i : Z_{i} \neq 0\}\]

  Since $z$ is a nonzero solution, $\kp$ is well defined. We conduct the proof as follows:
    \begin{enumerate}
      \item[(a)] Using the induction hypothesis, we show that the first nonzero $Z_i$ must occur for $i = \kp \leq 4^{d-1} \log n - (d-1)$.
      \item[(b)] The sequence $(Z_i)_i$ verifies $\frac{k-i}{i+d} Z_i \leq Z_{i+1}$.
      \item[(c)] Using the two previous results, we prove $\sum_{i=0}^k Z_i > 2n$.
    \end{enumerate}
  The contradiction comes then from $\sum_{i=0}^k Z_i \leq \sum_{e_1 + \dots + e_d \leq k} |z_{e_1,\dots,e_d}| \leq 2n$


  \textit{Proof of (a).} Assume $\kp > 0$ (otherwise the result is trivial). According to Equations (\ref{Eq:babaiGenEqBis}), and knowing that $z_{e_1,\dots,e_d} = 0$ whenever $e_1 + \dots + e_d < \kp$, we have
    \[\sum_{s=1}^d (e_s + 1) z_{e_1,\dots,e_{s-1},e_s+1,e_{s+1},\dots,e_d} = 0\]
  for all $e_1 + \dots + e_d = \kp - 1$. If we set apart the last term $z_{e_1,\dots,e_{d-1},e_d+1}$, we obtain
    \[(\kp - (e_1+\dots+e_{d-1}))z_{e_1,\dots,e_{d-1},e_d+1} + \sum\limits_{s=1}^{d-1} (e_s + 1) z_{e_1,\dots,e_{s-1},e_s+1,e_{s+1},\dots,e_d} = 0 \]

  Let $z'_{e_1,\dots,e_{d-1}} = z_{e_1,\dots,e_{d-1},\kp-(e_1 + \dots + e_{d-1})}$ for all $0 \leq e_1 + \dots + e_{d-1} \leq \kp$. We can change the variables in the previous equation as follows
    \[
      \left\{
        \def\arraystretch{1.5}
        \begin{array}{l l}
          (\kp-(e_1 + \dots + e_{d-1})) z'_{e_1,\dots,e_{d-1}} + \sum\limits_{s=1}^{d-1} (e_s + 1) z'_{e_1,\dots,e_{s-1},e_s+1,e_{s+1},\dots,e_{d-1}} = 0 \\
          0 \leq e_1 + \dots + e_{d-1} \leq \kp - 1
        \end{array}
      \right.
    \]
  This is equivalent to Equations (\ref{Eq:babaiGenEqBis}) at rank $d-1$. Moreover, $\sum\limits_{e_1 + \dots + e_{d-1} \leq \kp} |z'_{e_1,\dots,e_{d-1}}| \leq 2n$, and there exists $e_1 + \dots + e_d = \kp$ such that $z_{e_1,\dots,e_d} \neq 0$ (by definition of $\kp$), i.e. $z'_{e_1,\dots,e_{d-1}} \neq 0$. Consequently, it corresponds to a nonzero integral solution to Equations (\ref{Eq:babaiGenEqBis}) under Constraints (\ref{Eq:babaiGenConsBis}) at rank $d-1$ with parameter $\kp$. According to our induction hypothesis it implies $\kp \leq 4^{d-1} \log n - (d-1) $.


  \textit{Proof of (b).} Setting apart $z_{e_1,\dots,e_d}$ in Equations (\ref{Eq:babaiGenEqBis}), and using the triangle inequality, we obtain
    \[(k-(e_1 + \dots + e_d)) |z_{e_1,\dots,e_d}| \leq \sum\limits_{s=1}^d (e_s + 1) |z_{e_1,\dots,e_{s-1},e_s+1,e_{s+1},\dots,e_d}|\]
  for all $e_1 + \dots + e_d \leq k$. In particular, if we choose $e_1 + \dots + e_d$ such that $Z_{e_1+\dots+e_d} = |z_{e_1,\dots,e_d}|$ then
    \[
      \begin{split}
        (k-(e_1 + \dots + e_d)) Z_{e_1 + \dots + e_d} & \leq \sum_{s=1}^d (e_s + 1) |z_{e_1,\dots,e_{s-1},e_s+1,e_{s+1},\dots,e_d}| \\
                                                      & \leq \sum_{s=1}^d (e_s + 1) Z_{e_1 + \dots + e_d + 1} \\
                                                      & \leq (e_1 + \dots + e_d + d) Z_{e_1 + \dots + e_d + 1}
      \end{split}
    \]
  Thus $(k-i) Z_i \leq (i+d) Z_{i+1}$, where $i = e_1 + \dots + e_d$.


  \textit{Proof of (c).} Using (b), first note for $i > \kp$ that
    \[
      \begin{split}
        Z_i & \geq \frac{k - (i-1)}{(i-1)+d} \cdot \frac{k - (i-2)}{(i-2)+d} \cdots \frac{k - \kp}{\kp+d} \cdot Z_{\kp} \\
            & = \frac{(k - \kp)!}{(k-i)!} \cdot \frac{(\kp+d-1)!}{(i+d-1)!} \cdot Z_{\kp} \\
            & = \frac{(k+d-1)!}{(k-i)!(i+d-1)!} \cdot \frac{(k - \kp)!(\kp+d-1)!}{(k+d-1)!} \cdot Z_{\kp} \\
            & = \binom{k+d-1}{i+d-1} \binom{k+d-1}{\kp+d-1}^{-1} \cdot Z_{\kp} \\
            & \geq \binom{k+d-1}{i+d-1} \binom{k+d-1}{\kp+d-1}^{-1} \qquad \text{since $Z_{\kp} \geq 1$} \\
      \end{split}
    \]

  According to (a) and our initial hypothesis on $k$, we have $\kp+d-1 \leq 4^{d-1} \log n \leq (k+d-1)/4$. Thus $\sum_{i=\kp}^k \binom{k+d-1}{i+d-1} \geq \frac{1}{2} \cdot \sum_{i=0}^{k+d-1} \binom{k+d-1}{i} = 2^{k+d-2}$ and $\binom{k+d-1}{\kp+d-1}^{-1} \geq 2^{-(k+d-1)H(1/4)}$ (using the well-known bound $\binom{m}{\alpha m} \leq 2^{mH(\alpha)}$ where $H(\alpha) = - \log(\alpha^{\alpha} (1-\alpha)^{1-\alpha})$). Consequently, since $d \geq 2$ and $n \geq 2$, we obtain $\sum_{i=\kp}^k Z_i \geq 2^{(1-H(1/4))(k+d-1)-1} \geq 2^{(1-H(1/4)) 4^d \log n -1} > 2n$.
\end{proof}



\end{document}